\documentclass[11pt]{article} \usepackage{amssymb}
\usepackage{amsfonts} \usepackage{amsmath} \usepackage{amsthm} \usepackage{bm}
\usepackage{latexsym} \usepackage{epsfig}

\newtheorem*{theorem*}{Theorem}
\newtheorem{theorem}{Theorem}[section]

\newtheorem*{proposition*}{Proposition}

\newcommand{\ignore}[1]{}

\newcommand{\enote}[1]{} \newcommand{\knote}[1]{}
\newcommand{\rnote}[1]{}

\renewcommand{\P}[1]{{\mathbb{P}}\left[{#1}\right]}

\newcommand{\E}[1]{{\mathbb{E}}\left[{#1}\right]}

 \newcommand{\R}{\mathbb R}

\renewcommand{\phi}{\varphi}

\newcommand{\copyableTheorem}[2]{
 \newtheorem*{newthm#1}{Theorem \ref{thm#1}}
 \begin{newthm#1}
   {#2}
 \end{newthm#1}
 \expandafter\newcommand\expandafter{\csname thm#1\endcsname}{
   \begin{theorem}
     \label{thm#1}
     {#2}
   \end{theorem}
 }
}
\newcommand{\uu}[1]{\mathbb{U}\left[{#1}\right]}
\newcommand{\uup}[2]{\mathbb{U}_{#1}\left[{#2}\right]}
\newcommand{\geo}[1]{\mathbb{G}\left[{#1}\right]}

\begin{document}
\title{A lower bound on seller revenue in single buyer monopoly
  auctions}

\author{Omer Tamuz\footnote{Weizmann Institute, Rehovot 76100, Israel}}
\maketitle  

\begin{abstract}
  We consider a monopoly seller who optimally auctions a single object
  to a single potential buyer, with a known distribution of
  valuations. We show that a tight lower bound on the seller's
  expected revenue is $1/e$ times the geometric expectation of the
  buyer's valuation, and that this bound is uniquely achieved for the
  equal revenue distribution. We show also that when the valuation's
  expectation and geometric expectation are close, then the seller's
  expected revenue is close to the expected valuation.
\end{abstract}

\section{Introduction}
Consider a monopoly seller, selling a single object to a single
potential buyer. We assume that the buyer has a valuation for the
object which is unknown to the seller, and that the seller's
uncertainty is quantified by a probability distribution, from which it
believes the buyer picks its valuation.

Assuming that the seller wishes to maximize its expected revenue,
Myerson~\cite{myerson1981optimal} shows that the optimal incentive
compatible mechanism involves a simple one-time offer: the seller
(optimally) chooses a price and offers the buyer to buy the object for
this price; the assumption is that the buyer accepts the offer if its
valuation exceeds this price.  Myerson's seminal paper has become a
classical result in auction theory, with numerous follow-up studies. A
survey of this literature is beyond the scope of this paper (see,
e.g.,~\cite{krishna2009auction,klemperer1999auction}).

The expected seller revenue is an important, simple intrinsic
characteristic of the valuation distribution.  A natural question is
its relation with various other properties of the distribution. For
example, can seller revenue be bounded given such characterizations of
the valuation as its expectation, entropy, etc.?  An immediate upper
bound on seller revenue is the buyer's expected valuation. In fact,
the seller can extract the buyer's expected valuation only if the
seller knows the buyer's valuation exactly - i.e., the distribution
over valuations is a point mass.

Lower bounds on seller revenue are important in the study of
approximations to Myerson auctions (see., e.g., Hartline and
Roughgarden~\cite{hartline2009simple}, Daskalakis and
Pierrakos~\cite{daskalakis2011simple}). A general lower bound on the
seller's revenue is known when the distribution of the buyer's
valuation has a monotone hazard rate; in this case, the seller's
expected revenue is at least $1/e$ times the expected valuation (see
Hartline, Mirrokni and Sundararajan~\cite{hartline2008optimal}, as
well as Dhangwatnotai, Roughgarden and
Yan~\cite{dhangwatnotai2010revenue}).

This bound does not hold in general: as an extreme example, the equal
revenue distribution discussed below has infinite expectation but
finite seller revenue.  The family of monotone hazard rate
distributions does not include many important distributions such the
Pareto distribution or other power law distributions, or in fact any
distribution that doesn't have a very thin tail, vanishing at least
exponentially. The above mentioned lower bound for monotone hazard
rate distributions does not apply to these distributions, and indeed
it seems that the literature lacks any similar, general lower bounds
on seller revenue.

The {\em geometric expectation} of a positive random variable $X$ is
$\geo{X} = \exp(\E{\log X})$ (see,
e.g.,~\cite{paolella2006fundamental}).  We show that a general lower
bound on the seller's expected revenue is $1/e$ times the geometric
expectation of the valuation.  Equivalently, the (natural) logarithm
of the expected seller revenue is greater than or equal to the
expectation of the logarithm of the valuation, minus one.  This bound
holds for any distribution of positive valuations. Notably, the {\em
  regularity} condition, which often appears in the context of Myerson
auctions, is not required here. This result is a new and perhaps
unexpected connection between two natural properties of distributions:
the geometric expectation and expected seller revenue.

We show that this bound is tight in the following sense: for a fixed
value of the geometric mean, there is a unique cumulative distribution
function (CDF) of the buyer's valuation for which the bound is
achieved; this distribution is the equal revenue distribution, with
CDF of the form $F(v) = 1-c/v$ for $v \geq c$. This distribution is
``special'' in the context of single buyer Myerson auctions, as it is
the only one where seller revenue is identical for all prices.

The ratio between expected valuation and expected seller revenue is a
natural measure of the uncertainty of the valuation
distribution. Also, the discrepancy between the geometric expectation
and the (arithmetic) expectation of a positive random variable is a
well known measure of its dispersion. Hence, when the ratio between
the expectations is close to one, one would expect the amount of
uncertainty to be low and therefore seller revenue to be close to the
expected valuation. We show that this is indeed the case: when the
buyer's valuation has finite expectation, and the geometric
expectation is within a factor of $1-\delta$ of the expectation, then
seller revenue is within a factor of $1-2^{4/3}\delta^{1/3}$ of the
expected valuation. Similarly, it is easy to show that when the
variance of the valuation approaches zero then seller revenue also
approaches the expected valuation.

\section{Definitions and results}
We consider a seller who wishes to sell a single object to a single
potential buyer. The buyer has a valuation $V$ for the object which is
picked from a distribution with CDF $F$, i.e.  $F(v) = \P{V \leq v}$.

We assume that $V$ is positive, so that $\P{V \leq 0} = 0$ or $F(0) =
0$. We otherwise make no assumptions on the distribution of $V$; it
may be atomic or non-atomic, have or not have an expectation, etc.

The seller offers the object to the buyer for a fixed price $p$. The
buyer accepts the offer if $p < V$, in which case the seller's revenue
is $p$. Otherwise, i.e., if $p \geq V$, then the seller's revenue is
0. Thus, the seller's expected revenue for price $p$, which we denote
by $\uup{p}{V}$, is given by
\begin{align}
  \label{eq:u-p}
  \uup{p}{V} = p\P{p < V} = p(1-F(p)).
\end{align}
We define
\begin{align}
  \label{eq:u-s-max}
  \uu{V} = \sup_p\uup{p}{V} = \sup_p p(1-F(p)).
\end{align}
When this supremum is achieved for some price $p$ then $\uu{V}$ is the
seller's maximal expected revenue, achieved in the optimal Myerson
auction with price $p$.

We define the {\em geometric expectation} (see,
e.g.,~\cite{paolella2006fundamental}) of a positive real random
variable $X$ by $\geo{X} = \exp\left(\E{\log X}\right)$.  Note that
$\geo{X} \leq \E{X}$ by Jensen's inequality, and that equality is
achieved only for point mass distributions, i.e., when the buyer's
valuation is some fixed number. Note that likewise $\uu{V} \leq
\E{V}$, again with equality only for point mass distributions.

The equal revenue distribution
with parameter $c$  has the following CDF:
\begin{align}
  \label{eq:tight}
  \Phi_c(p) = \begin{cases}0&p \leq c\\ 1-\frac{c}{p}& p>c\end{cases}.
\end{align}
It is called ``equal revenue'' because if $V_c$ has CDF $\Phi_c$ then
$\uup{p}{V_c} = \uu{V_c}$ for all $p \geq c$.

Our main result is the following theorem.
\begin{theorem}\label{thmGeometricLowerBound}
  Let $V$ be a positive random variable. Then $\uu{V} \geq
  \frac{1}{e}\geo{V}$, with equality if and only if $V$ has the equal
  revenue CDF $\Phi_c$ with $c = \uu{V}$.
\end{theorem}
\begin{proof}
  Let $V$ be a positive random variable with CDF $F$.  By
  Eq.~\ref{eq:u-s-max} we have that
  \begin{align}
    \label{eq:basic}
    \log \uu{V} \geq \log p + \log(1-F(p))
  \end{align}
  for all $p$. We now take the expectation of both sides with respect
  to $p \sim F$:
  \begin{align}
    \label{eq:expectations}
    \int_0^\infty \log \uu{V} dF(p) \geq \int_0^\infty \log p \,dF(p) +
    \int_0^\infty\log(1-F(p)) dF(p).
  \end{align}
  Since $\uu{V}$ is a constant then the l.h.s.\ equals $\log \uu{V}$. The first
  addend on the r.h.s.\ is simply $\E{\log V}$. The second is $\E{\log
    (1-F(V))}$; note that $F(V)$ is distributed uniformly on $[0,1]$,
  and that therefore
  \begin{align*}
    \E{\log (1-F(V))} = \int_0^1\log(1-x)dx = -1.
  \end{align*}
  Hence Eq.~\ref{eq:expectations} becomes:
  \begin{align*}
    \log \uu{V} \geq \E{\log V} - 1,
  \end{align*}
  and
  \begin{align*}
    \uu{V} \geq \frac{1}{e}\exp(\E{\log V}) = \frac{1}{e}\geo{V}.
  \end{align*}

  To see that $\uu{V}=\frac{1}{e}\geo{V}$ only for the equal revenue
  distribution with parameter $\uu{V}$, note that we have equality in
  Eq.~\ref{eq:basic} for all $p$ in the support of $F$ if and only if
  $F=\Phi_c$ for some $c$, and that therefore we have equality in
  Eq.~\ref{eq:expectations} if and only if $F=\Phi_c$ for some
  $c$. Finally, a simple calculation yields that $c=\uu{V}$.
\end{proof}

Note that this proof in fact demonstrates a stronger statement, namely
that the expected revenue is at least $\frac{1}{e}\geo{V}$ for a
seller picking a random price from the distribution of
$V$. Dhangwatnotai, Roughgarden and
Yan~\cite{dhangwatnotai2010revenue} use similar ideas to show lower
bounds on revenue, for valuation distributions with monotone hazard
rates.

We next show that when the geometric expectation approaches the
(arithmetic) expectation then the seller revenue also approaches the
expectation.

\begin{theorem}
  Let $V$ be a positive random variable with finite expectation, and
  let $\geo{V}=(1-\delta)\E{V}$.  Then $\uu{V} \geq
  \left(1-2^{4/3}\delta^{1/3}\right)\E{V}$.
\end{theorem}
\begin{proof}
  Let $V$ be a positive random variable with finite expectation, and
  denote $1-\delta = \frac{\geo{V}}{\E{V}}$.  We normalize $V$ so that
  $\E{V} = 1$, and prove the claim by showing that $\uu{V} \geq
  1-2^{4/3}\delta^{1/3}$.  
  
  Consider the random variable $V-1-\log V$.  Since $\E{V}=1$, we have
  that $\E{V-1-\log V} = -\log \geo{V} = -\log(1-\delta)$. Since $x -
  1 \geq \log x$ for all $x>0$, then $V-1-\log V$ is
  non-negative. Hence by Markov's inequality
  \begin{align*}
    \P{V-1-\log V \geq -k \log (1-\delta)} \leq \frac{1}{k},
  \end{align*}
  or
  \begin{align}
    \label{eq:concentration}
    \P{Ve^{1-V} \leq (1-\delta)^k} \leq \frac{1}{k}.
  \end{align}
  This inequality is a concentration result, showing that when
  $\delta$ is small then $Ve^{1-V}$ is unlikely to be much less than
  one. However, for our end we require a concentration result on $V$
  rather than on $Ve^{1-V}$; that will enable us to show that the
  seller can sell with high probability for a price close to the
  arithmetic expectation. To this end, we will use the {\em Lambert
    $W$ function}, which is defined at $x$ as the solution of the
  equation $W(x)e^{W(x)} = x$.  We use it to solve the inequality of
  Eq.~\ref{eq:concentration} and arrive at
  \begin{align*}
    \P{V \leq -W\left(-(1-\delta)^k/e\right)} \leq \frac{1}{k},
  \end{align*}
  which is the concentration result we needed: $V$ is unlikely to be
  small when $\delta$ is small.  It follows that by setting the price
  at $-W\left(-(1-\delta)^k/e\right)$, the seller sells with
  probability at least $1-1/k$, and so
  \begin{align*}
    \uu{V} \geq -W\left(-(1-\delta)^k/e\right) \cdot \Big(1-1/k\Big).
  \end{align*}
  Now, an upper bound on $W$ is the
  following~\cite{corless1996lambertw}:
  \begin{align*}
    W(x) \leq -1+\sqrt{2(ex+1)},
  \end{align*}
  and so
  \begin{align*}
    \uu{V} \geq \Big(1-\sqrt{2(1-(1-\delta)^k}\Big) \cdot
    \Big(1-1/k\Big)
    \geq \Big(1-\sqrt{2\delta k}\Big) \cdot
    \Big(1-1/k\Big).
  \end{align*}
  Setting $k=(2\delta)^{-1/3}$ we get
  \begin{align*}
    \uu{V}
    \geq \Big(1-(2\delta)^{1/2} (2\delta)^{-1/6}\Big) \cdot \Big(1-(2\delta)^{1/3}\Big)    \geq 1-2(2\delta)^{1/3}.
  \end{align*}  
\end{proof}

\section{Open questions}
It may very well be possible to show tighter {\em upper} bounds for
$\uu{V}$, using continuous entropy. For example, let $V$ have
expectation $1$ and entropy at least $1$. Then $\uu{V}$ is
at most $1/e$: in fact, it is equal to $1/e$ since, by maximum entropy
arguments, there is only one distribution on $\R^+$ (the exponential
with expectation $1$) that satisfies both conditions, and for this
distribution $\uu{V}=1/e$. 

One could hope that it is likewise possible to prove upper bounds on
$\uu{V}$, given that $V$ has expectation $1$ and entropy at least $h <
1$; intuitively, the entropy constraint should force $V$ to spread
rather than concentrate around its expectation, decreasing the
seller's expected revenue.

\section{Acknowledgments}
We would like to thank Elchanan Mossel for commenting on a preliminary
version of this paper. We owe a debt of gratitude to the anonymous
reviewer who helped us improve the paper significantly through many
helpful suggestions.

This research is supported by ISF grant 1300/08, and by a Google
Europe Fellowship in Social Computing.

\bibliographystyle{elsarticle-num} \bibliography{bundling}

\end{document}